\documentclass[11pt]{amsart} 
\usepackage{amscd,amssymb,amsxtra}
\usepackage[mathscr]{eucal}
\usepackage{comment}
\usepackage{color}
\usepackage{enumerate}
\usepackage[colorlinks,citecolor=refkey]{hyperref}
\makeatletter
\let\@secnumfont\bfseries
\def\section{\@startsection{section}{1}%
  \z@{4\linespacing\@plus\linespacing}{\linespacing}%
  {\bfseries\centering}}
\def\introsection{\@startsection{section}{1}%
  \z@{3\linespacing\@plus\linespacing}{\linespacing}%
  {\bfseries\centering}}
\def\subsection{\@startsection{subsection}{2}%
   \z@{1.25\linespacing\@plus.7\linespacing}{.5\linespacing}%
   {\normalfont\bfseries}}
\def\subsectionsinline{\def\subsection{\@startsection{subsection}{2}%
  \z@{1\linespacing\@plus.7\linespacing}{-.5em}%
  {\normalfont\bfseries}}}

\makeatother

\theoremstyle{definition}

\newtheorem{example}[equation]{Example}

\newtheorem*{definition*}{Definition}
\newtheorem*{example*}{Example}
\newtheorem*{problem*}{Problem}
\newtheorem*{exercise*}{Exercise}
\newtheorem*{question*}{Question}
\newtheorem*{construction*}{Construction}

\theoremstyle{remark}

\newtheorem{remark}[equation]{Remark}

\newtheorem*{note*}{Note}
\newtheorem*{notation*}{Notation}
\newtheorem*{remark*}{Remark}

\theoremstyle{plain}
\newtheorem{theorem}[equation]{Theorem}
\newtheorem{corollary}[equation]{Corollary}
\newtheorem{lemma}[equation]{Lemma}

\newtheorem*{theorem*}{Theorem}
\newtheorem*{corollary*}{Corollary}
\newtheorem*{lemma*}{Lemma}
\newtheorem*{proposition*}{Proposition}
\newtheorem*{conjecture*}{Conjecture}
\newtheorem*{claim*}{Claim}
\newtheorem*{proposal*}{Proposal}
\newtheorem*{conclusion*}{Conclusion}
\newtheorem*{hypothesis*}{Hypothesis}

\numberwithin{equation}{section}

\definecolor{refkey}{rgb}{0,.6,.4}

\renewcommand{\:}{\colon}

\DeclareMathOperator{\Aut}{Aut}
\newcommand{\CC}{{\mathbb C}}
\newcommand{\CP}{{\mathbb C\mathbb P}}

\DeclareMathOperator{\End}{End}

\DeclareMathOperator{\Hom}{Hom}

\newcommand{\PP}{{\mathbb P}}

\newcommand{\RR}{{\mathbb R}}
\newcommand{\TT}{\mathbb T}

\DeclareMathOperator{\Tr}{Tr}

\newcommand{\ZZ}{{\mathbb Z}}

\newcommand{\chiup}{\raise.5ex\hbox{$\chi$}}

\newcommand{\mstrut}{^{\vphantom{1*\prime y\vee M}}}

\newcommand{\res}[1]{\negmedspace\bigm|\mstrut_{#1}}
\newcommand{\temsquare}{\raise3.5pt\hbox{\boxed{ }}}

\newcommand{\zmod}[1]{\ZZ/#1\ZZ}

\newcommand{\zt}{\zmod2}

\DeclareMathOperator{\Pin}{Pin}
\newcommand{\CHom}{\Hom_{\CC}}
\newcommand{\GH}{G(\sH)}
\newcommand{\Gf}{\Gamma _f}
\newcommand{\Lp}{L^{\perp}}
\newcommand{\PH}{\PP\sH}
\newcommand{\PV}{\PP V}
\newcommand{\PU}{\PP U}
\newcommand{\QAut}{\Aut_{\textnormal{qtm}}}
\newcommand{\hf}{\hat{\phi }}

\newcommand{\qsym}{\QAut(\PH)} 
\newcommand{\sH}{\mathscr{H}}
\newcommand{\sL}{\mathcal{L}}

\begin{document}

\abovedisplayskip18pt plus4.5pt minus9pt
\belowdisplayskip \abovedisplayskip
\abovedisplayshortskip0pt plus4.5pt
\belowdisplayshortskip10.5pt plus4.5pt minus6pt
\baselineskip=15 truept
\marginparwidth=55pt

\renewcommand{\labelenumi}{\textnormal{(\roman{enumi})}}




 \title{On Wigner's Theorem} 
 \author[D. S. Freed]{Daniel S.~Freed}
 \thanks{The author is supported by the National Science Foundation under
grant DMS-0603964}
 \address{Department of Mathematics \\ University of Texas \\ 1 University
Station C1200\\ Austin, TX 78712-0257}
 \email{dafr@math.utexas.edu}
 \dedicatory{For Mike Freedman, on the occasion of his $60^{\textnormal{th}}$
birthday}
 \date{July 31, 2012}
  \begin{abstract} 
 Wigner's theorem asserts that any symmetry of a quantum system is unitary or
antiunitary.  In this short note we give two proofs based on the geometry of
the Fubini-Study metric.
  \end{abstract}
\maketitle


The space of pure states of a quantum mechanical system is the projective
space~$\PH$ of lines in a separable complex Hilbert space~$\bigl(\sH,\langle
-,- \rangle\bigr)$, which may be finite or infinite dimensional.  It carries
a symmetric function $p\:\PH\times \PH\to[0,1]$ whose value~$p(L_1,L_2)$ on
states $L_1,L_2\in \PH$ is the \emph{transition probability}: if $\psi _i\in
L_i$ is a unit norm vector in the line~$L_i$, then
  \begin{equation*}
     p(L_1,L_2) = |\langle \psi _1,\psi _2 \rangle|^2. 
  \end{equation*}
Let $\QAut(\PH)$~denote the group of symmetries of~$(\PH,p)$, the group of
quantum symmetries.  A fundamental theorem of Wigner\footnote{As I learned
in~\cite[p.~74]{Bo}, this theorem was first asserted in a 1928~joint
paper~\cite[p.~207]{VNW} of von Neumann and Wigner, though with only a brief
justification.  A more complete account appeared in Wigner's book (in the
original German) in~1931.}~\cite[\S20A, \S26]{Wi}, \cite{Ba}, \cite[\S2A]{We}
expresses~$\QAut(\PH)$ as a quotient of linear and antilinear symmetries of
~$\sH$.  This note began with the rediscovery of a formula which relates the
quantum geometry of~$(\PH,p)$ to a more familiar structure in differential
geometry: the Fubini-Study K\"ahler metric on~$\PH$.  It leads to two proofs
of Wigner's theorem, Theorem~\ref{thm:5} of this note, based on the
differential geometry of projective space.

The proofs here use more geometry than the elementary proofs~\cite{Ba},
\cite[\S2A]{We}.  We take this opportunity to draw attention to Wigner's
theorem and to the connection between quantum mechanics and projective
geometry.  It is a fitting link for a small tribute to Mike Freedman, whose
dual careers in topology and condensed matter physics continue to inspire.

\bigskip\bigskip\medskip
 Let $d\:\PH\times \PH\to\RR^{\ge0}$ be the distance function associated to
the Fubini-Study metric.

  \begin{theorem}[]\label{thm:1}
 The functions~$p$ and~$d$ are related by
  \begin{equation}\label{eq:2}
     \cos(d)=2p-1. 
  \end{equation}
  \end{theorem}

\noindent
 As a gateway into the literature on `geometric quantum mechanics',
where~\eqref{eq:2} can be found,\footnote{Notice that \eqref{eq:2} is
equivalent to~$p=\cos^2(d/2)$. } see~\cite{BH} and the references therein.

  \begin{corollary}[]\label{thm:2}
 $\QAut(\PH)$ is the group of isometries of~$\PH$ with the Fubini-Study
distance function.
  \end{corollary}

  \begin{remark}[]\label{thm:3}
 If $\sH$~is infinite dimensional, then $\PH$~is an infinite dimensional
smooth manifold modeled on a Hilbert space.  Basic notions of calculus and
differential geometry carry over to Hilbert manifolds~\cite{L}.  The
Myers-Steenrod theorem asserts that a distance-preserving map between two
Riemannian manifolds is smooth and preserves the Riemannian metric.  That
theorem is also true on Riemannian manifolds modeled on Hilbert
manifolds~\cite{GJR}.\footnote{The proof depends on the existence of geodesic
convex neighborhoods, proved in~\cite[\S VIII.5]{L}.  For the Fubini-Study
metric on~$\PH$ such neighborhoods may easily be constructed explicitly.  I
thank Karl-Hermann Neeb for his inquiry about the Myers-Steenrod theorem in
infinite dimensions.}  So in the sequel we use that a distance-preserving map
$\phi \:\PH\to \PH$ is smooth and is an isometry in the sense of Riemannian
geometry.
  \end{remark}
 
The tangent space to~$\PH$ at a line~$L\subset \sH$ is canonically $T\mstrut
_L\PH\cong \CHom(L,\Lp)$, where $\Lp\subset \sH$ is the orthogonal complement
to~$L$, a closed subspace and therefore itself a Hilbert space.  If
$f_1,f_2\:L\to\Lp$, then the Fubini-Study hermitian metric is defined by
  \begin{equation}\label{eq:3}
     \langle f_1,f_2 \rangle = \Tr(f_1^*f\mstrut _2). 
  \end{equation}
The adjoint~$f_1^*$ is computed using the inner products on~$L$ and~$\Lp$.
The composition~$f_1^*f\mstrut _2$ is an endomorphism of~$L$, hence
multiplication by a complex number which we identify as the trace of the
endomorphism.  If $\ell \in L$ has unit norm, then the map
  \begin{equation}\label{eq:4}
     \begin{aligned} \CHom(L,\Lp) &\longrightarrow \Lp \\ f&\longmapsto
      f(\ell )\end{aligned} 
  \end{equation}
is a linear isometry for the induced metric on~$\Lp\subset \sH$.  The
underlying Riemannian metric is the real part of the hermitian
metric~\eqref{eq:3}; it only depends on the real part of the inner product
on~$\sH$.

  \begin{proof}[Proof of Theorem~\ref{thm:1}]
 Equation~\eqref{eq:2} is obvious on the diagonal in~$\PH\times \PH$, as well
as if $\dim\sH=1$.  Henceforth we rule out both possibilities.  Fix $L_1\not=
L_2\in \PH$ and let $V$~be the 2-dimensional space~$L_1+L_2\subset \sH$.  The
unitary automorphism of $\sH=V\oplus V^\perp$ which is~$+1$ on~$V$ and~$-1$
on~$V^\perp$ induces an isometry of~$\PH$ which has~$\PV$ as a component of
its fixed point set.  It follows that $\PV$~is totally geodesic.  Therefore,
to compute~$d(L_1,L_2)$ we are reduced to the case of the complex projective
line with its Fubini-Study metric: the round 2-sphere.
 
Let $e_1\in L_1$ have unit norm and choose~$e_2\in V$ to fill out a unitary
basis~$\{e_1,e_2\}$.  Then $\lambda e_1+e_2\in L_2$ for a unique $\lambda \in
\CC$.  If $\lambda =0$ then it is easy to check that~$d=\pi $ and~$p=0$,
consistent with~\eqref{eq:2}, so we now assume~$\lambda \not= 0$.
Identify~$\PV\setminus \{\CC\cdot e_2\}\approx \CC$ by $\CC\cdot (e_1+\mu
e_2)\leftrightarrow \mu $.  Use stereographic projection from the north
pole~$(1,0)$ in Euclidean 3-space $\RR\times \CC$ to identify~$\{0\}\times
\CC\approx S^2\setminus \{(1,0)\}$, where $S^2\subset \RR\times \CC$ is the
unit sphere.  Under these identifications we have
  \begin{equation*}
     \begin{aligned} L_1 &\longleftrightarrow \bigl(-1\,,\,0 \bigr) \\ L_2
      &\longleftrightarrow \bigl(-\,\frac{|\lambda |^2-1}{|\lambda |^2+1}\,,\,
      \frac{2|\lambda |^2}{|\lambda |^2+1}\,\frac 1\lambda
      \bigr)\end{aligned} 
  \end{equation*}
from which $\cos(d) = (|\lambda |^2-1)/(|\lambda |^2+1)$ can be computed as
the inner product of vectors in the 3-dimensional vector space~$\RR\oplus
\CC$.  Since $p=|\lambda |^2/(|\lambda |^2+1)$, equation~\eqref{eq:2} is
satisfied.
  \end{proof}
 
\bigskip
 A \emph{real} linear map $S\:\sH\to\sH$ is \emph{antiunitary} if it is
conjugate linear and  
  \begin{equation*}
     \langle S\psi _1,S\psi _2 \rangle = \overline{\langle \psi _1,\psi _2
     \rangle}\qquad \textnormal{for all $\psi _1,\psi _2\in \sH$}. 
  \end{equation*}
Let $\GH$~denote the group consisting of all unitary and antiunitary
operators on~$\sH$.  In the norm topology it is a Banach Lie group~\cite{M}
with two contractible components; the same is true in the compact-open
topology~\cite[Appendix~D]{FM}.  The identity component is the group~$U(\sH)$ of
unitary transformations.  Any $S\in \GH$ maps complex lines to complex lines,
so induces a diffeomorphism of~$\PH$, and since $S$~preserves the real part
of~$\langle -,- \rangle$ the induced diffeomorphism is an isometry.  The unit
norm scalars $\TT\subset \GH$ act trivially on~$\PH$, so there is an
exact\footnote{We assume $\dim\sH>1$.}  sequence of Lie groups
  \begin{equation}\label{eq:7}
     1\longrightarrow \TT\longrightarrow \GH\longrightarrow
     \qsym.
  \end{equation}
Note that $\TT$~is not central since antiunitary maps conjugate scalars.

  \begin{theorem}[Wigner~\cite{Wi}]\label{thm:5}
 The homomorphism $\GH\to\qsym$ is surjective: every quantum symmetry
of~$\PH$ lifts to a unitary or antiunitary operator on~$\sH$.
  \end{theorem}

\noindent 
 By Corollary~\ref{thm:2} the same is true for isometries of the Fubini-Study
metric, and indeed we prove Wigner's Theorem by computing the group of
isometries.

  \begin{remark}[]\label{thm:4}
 If $\rho \:G\to\qsym$ is any group of quantum symmetries, then the
surjectivity of $\GH\to\qsym$ implies the extension~\eqref{eq:7} pulls back
to a twisted central extension of~$G$.  The twist is the homomorphism
$G\to\zt$ which tells whether a symmetry lifts to be unitary or antiunitary.
The isomorphism class of this twisted central extension is then an invariant
of~$\rho $.  This is the starting point for joint work with Greg
Moore~\cite{FM} about quantum symmetry classes and topological phases in
condensed matter physics.
  \end{remark}

  \begin{example}[]\label{thm:7}
 $\PP(\CC^2)=\CP^1$~ with the Fubini-Study metric is the round 2-sphere of
unit radius.  Its isometry group is the group~$O(3)$ of orthogonal
transformations of~$SO(3)$.  The identity component~$SO(3)$ is the image of
the group~$U(2)$ of unitary transformations of~$\CC^2$.  The other component
of~$O(3)$ consists of orientation-reversing orthogonal transformations, such
as reflections, and they lift to antiunitary symmetries of~$\CC^2$.  In this
case the group~$G(\sH)$ is also known as~$\Pin^c(3)$; see~\cite{ABS}.
  \end{example}

\bigskip
 We present two proofs of Theorem~\ref{thm:5}.  The first is based on the
following standard fact in Riemannian geometry.

  \begin{lemma}[]\label{thm:8}
 Let $M$~be a Riemannian manifold, $p\in M$, and $\phi \:M\to M$ an isometry
with~$\phi (p)=p$.  Suppose $B_r\subset T_pM$ is the open ball of radius~$r$
centered at the origin and assume the Riemannian exponential map~$\exp_p$
maps~$B_r$ diffeomorphically into~$M$.  Then in exponential coordinates $\phi
\res{B_r}$~ equals the restriction of the linear isometry~$d\phi _p$
to~$B_r$. 
  \end{lemma}

  \begin{proof}
 If $\xi \in B_r$, then $\exp_p(\xi )=\gamma _\xi (1)$, where $\gamma _\xi
\:[0,1]\to M$ is the unique geodesic which satisfies $\gamma _\xi
(0)=p,\;\dot\gamma _\xi (0)=\xi $.  Since $\phi $~maps geodesics to
geodesics, $\phi \circ \exp_p = \exp_p\circ d\phi _p$ on~$B_r$, as desired.
  \end{proof}

\noindent
 If $\rho \:[0,r')\to[0,r)$ is a diffeomorphism for some~$r'>0$, then 
  \begin{equation}\label{eq:8}
     \xi  \longmapsto \exp_p\bigl(\rho (|\xi |)\xi \bigr) 
  \end{equation}
maps~$B_{r'}$ diffeomorphically into~$M$, and $\phi $~in this coordinate
system is also linear. 

  \begin{proof}[First Proof of Theorem~\ref{thm:5}]
 Let $\phi \:\PH\to\PH$ be an isometry.  Composing with an isometry in~$\GH$
we may assume $\phi (L)=L$ for some~$L\in \PH$.  The tangent space~$T_L\PH$
is canonically~$\CHom(L,\Lp)$, and also $f\in \CHom(L,\Lp)$ determines
$\Gf\in \PH$ by $\Gf\subset \sH=L\oplus \Lp$ is the graph of~$f$.  We claim
$f\mapsto\Gf$ has the form \eqref{eq:8} for some $\rho \:[0,\infty )\to[0,\pi
)$.  It suffices to show that for any $f\in \CHom(L,\Lp)$ of unit norm, the
map $t\mapsto \Gamma _{tf}$ traces out a (reparametrized) geodesic in a
parametrization independent of~$f$.  As in the proof of Theorem~\ref{thm:1}
this reduces to $\dim\sH=2$ and so to an obvious statement about the round
2-sphere.  It follows from Lemma~\ref{thm:8} that $\phi $~is a \emph{real}
isometry $S\in \End_{\RR}\bigl(\CHom(L,\Lp) \bigr)$.  It remains to prove
that $S$~is complex linear or antilinear; then we extend~$S$ by the identity
on~$L$ to obtain a unitary or antiunitary operator on $\sH=L\oplus \Lp$.
 
If $\dim\sH=2$ then Theorem~\ref{thm:5} can be verified (see
Example~\ref{thm:7}), so assume $\dim\sH>2$.  Identify $\CHom(L,\Lp)\approx
\Lp$ as in~\eqref{eq:4}.  Since $S\in \End_{\RR}(\Lp)$ maps complex lines
in~$\Lp$ to complex lines, there is a function $\alpha \:\Lp\setminus
\{0\}\to\CC$ such that $S(i\xi )=\alpha (\xi )S(\xi )$ for all nonzero~$\xi
\in \Lp$.  Fix $\xi \not= 0$ and choose~$\eta \in \Lp$ which is linearly
independent.  Then
  \begin{equation*}
     \begin{split} S\bigl(i(\xi +\eta ) \bigr)&= \alpha (\xi +\eta
      )\bigl[S(\xi )+S(\eta ) \bigr] \\ &= \alpha (\xi )S(\xi ) + \alpha
      (\eta )S(\eta )\end{split} 
  \end{equation*}
from which $\alpha (\xi )=\alpha (\eta )$.  Applied to~$i\xi ,\eta $ we learn
$\alpha (\xi )=\alpha (i\xi )$.  On the other hand, 
  \begin{equation*}
     -S(\xi )=S(-\xi )=\alpha (i\xi )S(i\xi ) = \alpha (i\xi )\alpha (\xi
     )S(\xi ), 
  \end{equation*}
whence $\alpha (\xi )^2=-1$.  By continuity either $\alpha \equiv i$
or~$\alpha \equiv -i$, which proves that $S$~is linear or $S$~is antilinear. 
  \end{proof}

\bigskip

The second proof leans on complex geometry. 

  \begin{lemma}[]\label{thm:9}
 An isometry $\phi \:\PH\to\PH$ is either holomorphic or antiholomorphic. 
  \end{lemma}

  \begin{proof}
 Let $I\:T\PH\to T\PH$ be the (almost) complex structure.  Then $I$~is
parallel with respect to the Levi-Civita covariant derivative, since $\PH$~is
K\"ahler, and so therefore is~$\phi ^*I$.  We claim \emph{any} parallel
almost complex structure~$J$ equals~$\pm I$; the lemma follows immediately.
 
If $J$~is parallel, then it commutes with the Riemann curvature tensor~$R$.
Compute at~$L\in \PH$ and identify $T_L\PH\approx \Lp$, as in~\eqref{eq:4}.
Then if $\xi ,\eta \in \Lp$ and $\langle \xi ,\eta \rangle=0$, since
$\PP(L\,\oplus\, \CC\cdot \xi \,\oplus\, \CC\cdot \eta )\subset \PH$ is
totally geodesic and has constant holomorphic sectional curvature
one~\cite[\S IX.7]{KN}, we compute
  \begin{equation*}
     \begin{aligned} R(\xi ,I\xi )\xi &=-|\xi |^2I\xi , \\ R(\xi ,I\xi
      )\eta &=-\frac12|\xi |^2I\eta .\end{aligned} 
  \end{equation*}
It follows that $J$~preserves every complex line~$K=\CC\cdot \xi \subset \Lp$
and commutes with~$I$ on~$K$.  Therefore, $J=\pm I$ on~$K$.  By continuity,
the sign is independent of~$K$ and~$L$.
  \end{proof}

  \begin{proof}[Second Proof of Theorem~\ref{thm:5}]
 First, recall that if $U$~is finite dimensional, then every holomorphic
symmetry of~$\PU$ is linear.  The proof is as follows.  Let $\sL\to\PU$ be
the canonical holomorphic line bundle whose fiber at~$L\in \PU$ is~$L$.  A
holomorphic line bundle on~$\PU$ is determined by its Chern class, so $\phi
^*\sL\cong \sL$.  Fix an isomorphism; it is unique up to scale.  There is an
induced linear map on the space $H^0(\PU;\sL^*)\cong U^*$ of global
holomorphic sections:
  \begin{equation}\label{eq:a}
     \phi ^*\:H^0(\PU;\sL^*)\longrightarrow H^0(\PU;\phi ^*\sL^*)\cong
     H^0(\PU;\sL^*).
  \end{equation}
The transpose~$\hf$ of~\eqref{eq:a} is the desired linear lift of~$\phi $.
 
Let $\phi \:\PH\to\PH$ be an isometry.  After composition with an element
of~$G(\sH)$ we may, by Lemma~\ref{thm:9}, assume $\phi $~is holomorphic and
fixes some~$L\in \PH$.  Let $U\subset \sH$ be a finite dimensional subspace
containing~$L$.  Then the pullback of $\sL_{\sH}\to\PH$ to~$\phi
^*\sL_{\sH}\res{\PP U}\to\PP U$ has degree one, so is isomorphic
to~$\sL_U\to\PP U$, and there is a unique isomorphism which is the identity
on the fiber over~$L$.  A functional~$\alpha \in \sH^*$ restricts to a
holomorphic section of~$\phi ^*\sL_{\sH}^*\res{\PP U}\to\PU$, so by
composition with the isomorphism~$\phi ^*\sL_{\sH}^*\res{\PP U}\cong \sL_U^*$
to an element of~$U^*$.  The resulting map~$\sH^*\to U^*$ is linear, and its
transpose $\hf\:U\to\sH$ is the identity on~$L$.  Let $U$~run over all finite
dimensional subspaces of~$\sH$ to define $\hf\:\sH\to\sH$.  The uniqueness of
the isomorphism~$\phi ^*\sL_{\sH}\res{\PP U}\cong \sL_U$ implies that
$\hf$~is well-defined and a linear lift of~$\phi $.  It is unitary since
$\phi $~is an isometry.
  \end{proof}

\providecommand{\bysame}{\leavevmode\hbox to3em{\hrulefill}\thinspace}
\providecommand{\MR}{\relax\ifhmode\unskip\space\fi MR }
\providecommand{\MRhref}[2]{%
  \href{http://www.ams.org/mathscinet-getitem?mr=#1}{#2}
}
\providecommand{\href}[2]{#2}

  \end{document}